\documentclass[11pt,oneside]{amsart}
\usepackage[T1]{fontenc}
\usepackage[latin9]{inputenc}
\usepackage{geometry}
\usepackage{verbatim}
\usepackage{textcomp}
\usepackage{amsthm}
\usepackage{amstext}
\usepackage{amssymb}
\usepackage{esint}

\makeatletter
\numberwithin{equation}{section}
\numberwithin{figure}{section}
\theoremstyle{plain}
\newtheorem{thm}{\protect\theoremname}
  \theoremstyle{definition}
  \newtheorem{defn}[thm]{\protect\definitionname}
  \theoremstyle{plain}
  \newtheorem{prop}[thm]{\protect\propositionname}
  \theoremstyle{plain}
  \newtheorem{lem}[thm]{\protect\lemmaname}
  \theoremstyle{remark}
  \newtheorem{rem}[thm]{\protect\remarkname}
  \theoremstyle{definition}
  \newtheorem{example}[thm]{\protect\examplename}

\def\R{{\mathbb R}}
\def\V{{\mathbb R}^d}
\def\Rep{{\bf Rep}}
\def\pol{{\bf Pol}}
\def\id{{\bf id}}
\def\Diff{{\bf Diff}}
\def\DiffB{\Diff_b}
\def\DiffBV{\Diff_{b,v}}

\def\bR{{\mathbb R}}

\def\bC{{\mathbb C}}

\def\bN{{\mathbb N}}

\def\spm{\star}
\def\Op{\mathrm{Op}}
\def\op{\mathrm{Op}_1}

\def\fio{\mathrm{FIO}}

\def\bp{\overline{\xi}}
\def\bq{\overline{x}}
\def\tp{\tilde{\alpha}}

\def\tq{\tilde{x}}
\def\tp{\tilde{\xi}}
\def\m{\frac{d\overline{\xi}\:d\overline{x}}{(2\pi\hbar)^d}}
\def\n{\frac{d\tilde{\xi}\:d\tilde{x}}{(2\pi\hbar)^d}}

\def\tp{\tilde{\xi}}
\def\tq{\tilde{x}}

\makeatother

\usepackage[english]{babel}
  \providecommand{\definitionname}{Definition}
  \providecommand{\examplename}{Example}
  \providecommand{\lemmaname}{Lemma}
  \providecommand{\propositionname}{Proposition}
  \providecommand{\remarkname}{Remark}
\providecommand{\theoremname}{Theorem}

\begin{document}

\selectlanguage{english}

\title{quantization of (volume-preserving) actions on $\bR^{d}$}

\begin{abstract}
We associate a space of (formal) representations on $C^{\infty}(\R^{d})[[\hbar]]$
(which we call quantizations) with an action of a group on $\R^{d}$
by smooth diffeomorphisms. If the action is further volume preserving,
these quantizations can be realized as unitary representations on
$L^{2}(\R^{d})$ by bounded $\hbar$-dependent Fourier integral operators,
the formal case corresponding to the asymptotics in the limit $\hbar\rightarrow0$.
We construct $DGA$s controlling these quantizations and prove existence
and rigidity results for them.
\end{abstract}

\author{Benoit Dherin and Igor Mencattini}

\address{Benoit Dherin and Igor Mencattini, ICMC-USP Universidade de Sao Paulo,
Avenida Trabalhador Sao-carlense 400 Centro, CEP: 13566-590, Sao Carlos,
SP, Brazil}

\maketitle
\tableofcontents{}

\section{Introduction}

The question of quantizing a smooth action $\varphi$ of a Lie group
$G$ on a manifold $M$ has received different (although related)
answers depending on the particular structures at hand on the manifold,
the type of Lie groups acting, the type of actions, and the quantization
theory used. 

When the manifold is symplectic, and the action is hamiltonian and
admits a momentum map, both geometric quantization theory (see \cite{GS,Ko}
for instance) and deformation quantization theory (see \cite{A2,B1,BFFLS,Xu,Rieffel2}
for instance) have their own notion of quantization. On the other
hand, Rieffel in \cite{Rieffel}, using ideas of deformation quantization,
introduced a notion of action quantization that supposes no symplectic
structure on the manifold to begin with if the group acting is $\bR^{d}$.
This program has been extended to various other groups and cases (see
\cite{B1,B2,LW}).

In this paper, we propose a quantization scheme for a general action
$\varphi$ of a group $G$ (not necessarily a Lie group) on $\bR^{d}$
by smooth diffeomorphisms. More precisely, we associate to such an
action a space $\Rep_{\varphi}(G)$ of representations (which we call
\textit{quantizations}) by certain formal operators on the space $C^{\infty}(\R^{d})[[\hbar]]$
of formal power series in $\hbar$ with coefficients in the smooth
functions on $\R^{d}$. In particular, the \textit{trivial quantization}\textbf{,
}obtained by the pullback of functions, $T_{g}\psi(x)=\psi(\varphi_{g}^{-1}(x))$
with $g\in G$, is always in $\Rep_{\varphi}(G)$, and the other representations
in $\Rep_{\varphi}(G)$ can be seen as {}``deformations'' of this
trivial quantization. 

The main result of this paper (Theorem \ref{thm:ExistenceAndRigidity})
gives cohomological obstructions to the existence of such {}``deformations''
as well as information with regards to their rigidity (i.e. when all
the quantizations in $\Rep_{\varphi}(G)$ are equivalent to the trivial
one). The main ingredient to prove these existence and rigidity results
is a Differential Graded Algebra (DGA) whose Maurer-Cartan elements
are in one-to-one correspondence with the quantizations of the action. 

When the action $\varphi$ is further volume preserving and bounded
(which means that $|\varphi_{g}'(x)|=1$ for all $g\in G$ and $x\in\bR^{d}$
with the additional condition that $\varphi_{g}$ and all of its derivatives
are bounded for all $g\in G$), $\Rep_{\varphi}(G)$ can be realized
as a space of unitary representations on $L^{2}(\R^{d})$ by certain
bounded Fourier Integral Operators, or FIOs for short (see \cite{D,EZ,Ho,M}
for general references), which depend on a parameter $\hbar$. In
this non-formal setting, there is also a DGA controlling quantization. 

Actually, the formal quantizations associated with an action by smooth
diffeomorphisms are constructed by taking the asymptotic expansion
in the limit $\hbar\rightarrow0$ of the FIOs used in the volume preserving
case and forgetting that these expansions come from honest bounded
operators. What results is a set of formal operators of infinite order,
which may not be {}``resummable'' if the action we start with is
not bounded. 

We also explain how geometric quantization (Example \ref{sub:Geometric-quantization-intermezz})
and deformation quantization (Section \ref{sub:Formal-G-systems})
are related to our quantization scheme for actions.

This paper is organized as follows:

In Section \ref{sec:Preliminaries}, we introduce the class of $\hbar$-dependent
FIOs we use to quantize volume-preserving actions. These operators
are of the form
\[
\Op(a,\varphi)\psi(x)=\int\psi(\overline{x})a(x,\overline{\xi})e^{\frac{i}{\hbar}\langle\overline{\xi},\varphi^{-1}(x)-\overline{x}\rangle}\frac{d\overline{\xi}d\overline{x}}{(2\pi\hbar)^{d}},
\]
where $\varphi$ is a (bounded) diffeomorphism of $\R^{d}$. We give
results on the continuity of these operators as well as their asymptotics
in the limit $\hbar\rightarrow0$, which we interpret as formal operators
of infinite order.

In Section \ref{sec:Quantization-of--actions}, we introduce the space
$\Rep_{\varphi}(G)$ of quantizations associated with an action together
with their corresponding $G$-\textit{systems.} The starting point
is the observation that the trivial quantization can be rewritten
in terms of the FIOs of the previous section as follows:
\[
T_{g}\psi=\Op(1,\varphi_{g})\psi.
\]
When the action is bounded, a $G$-system is a system $\{a_{g}(x,\xi)\}_{g\in G}$
of amplitudes such that the operators
\begin{equation}
T_{g}^{a}\psi=\Op(a_{g},\varphi_{g})\psi\label{eq:deformation}
\end{equation}
form a representation of $G$ on $L^{2}(\R^{d})$ by bounded operators.
We explain that the asymptotic expansion of these quantizations yields
a notion of formal $G$-systems and formal quantizations that can
be used when the action is no longer bounded. When the action is further
volume preserving, we can require the $G$-system to be so that the
corresponding representations are unitary. There are a number of examples
of this in the literature, but, mostly, when the amplitudes of the
$G$-system do not depend on $\xi$. Because of this, we conclude
this section by a study of these special $G$-systems, yielding to
Theorem \ref{l27}, which is an analog of our main result for formal
$G$-systems (Theorem \ref{thm:ExistenceAndRigidity}) in this special
case. 

In Section \ref{sec:The-complex-of}, we construct two DGAs controlling,
respectively, $G$-systems and their formal versions. We show that
Maurer-Cartan elements are in one-to-one correspondence with $G$-systems
(both in the formal and non-formal case) and that gauge equivalent
Maurer-Cartan elements give equivalent quantizations.

In Section \ref{sec:Existence-and-rigidity}, we state and prove the
main theorem of this paper (Theorem \ref{thm:ExistenceAndRigidity}),
which gives cohomological conditions with regards to the existence
and rigidity of formal $G$-systems. We spell out this theorem in
the case the action we start with is trivial, obtaining results (Theorem
\ref{thm:trivial-action}) very close to those of Pinzcon \cite{Pinzcon}
on deformations of representations.

\subsection*{Acknowledgments}

We thank Alberto Cattaneo, Ugo Bruzzo, Giuseppe Dito, Gianni Landi,
Marc Rieffel, Mauro Spreafico, Ali Tahzibi, Alan Weinstein, and Sergio
Zani for useful feedback and for pointing us toward related works,
as well as the hospitality of UC Berkeley and SISSA, where part of
this project was conducted. B.D. acknowledges support from FAPESP
grant 2010/15069-8 and 2010/19365-0 and the University of S\~ao Paulo.

\section{Preliminaries\label{sec:Preliminaries}}

In this section, we review a class of $\hbar$-dependent Fourier integral
operators that we will use in Section \ref{sec:Quantization-of--actions}
for action quantization purposes. We discuss the continuity of these
operators as well as the closeness of their composition. We also give
the asymptotics of these operators in the limit $\hbar\rightarrow0$,
which we will use later on to define a notion of {}``formal quantization''
of actions. Along the way, we review some facts about pseudo-differential
operators.

Throughout this paper, we will consider $\bR^{d}$ with its canonical
coordinates $x=(x_{1},\dots,x_{d})$, and we will identify its cotangent
bundle $T^{*}\bR^{d}$ with $\bR^{2d}=\bR^{d}\times(\bR^{d})^{*}$,
where $(\bR^{d})^{*}$ is the dual to $\R^{d}$ with dual coordinates
$\xi=(\xi_{1},\dots,\xi_{d})$. The paring between $\V$ and $(\V)^{\ast}$,
will be denoted by $\langle\cdot,\cdot\rangle$ so that $\langle x,\xi\rangle=\sum_{i=1}^{d}x_{i}\xi_{i}$.
Also, $\m$ will stand for the Lebesgue measure on $T^{*}\V$. We
will also make use of the multi-index notation: For $\alpha\in\bN^{d}$,
we define
\begin{gather*}
|\alpha|=\alpha_{1}+\cdots+\alpha_{d},\qquad y^{\alpha}=y^{\alpha_{1}}\cdots y^{\alpha_{d}}\quad\textrm{for }y\in\R^{d},\\
\partial_{x}^{\alpha}=\frac{\partial^{|\alpha|}}{\partial x^{\alpha_{1}}\cdots\partial x^{\alpha_{d}}},\qquad\partial_{\xi}^{\alpha}=\frac{\partial^{|\alpha|}}{\partial\xi^{\alpha_{1}}\cdots\partial\xi^{\alpha_{d}}},\\
D_{x}^{\alpha}=\frac{1}{i^{|\alpha|}}\partial_{x}^{\alpha},\qquad D_{\xi}^{\alpha}=\frac{1}{i^{|\alpha|}}\partial_{\xi}^{\alpha}.
\end{gather*}

\subsection{Fourier integral operators}

A \textbf{Fourier Integral Operator} (or FIO) on $\R^{d}$ is an integral
operator, denoted by $\Op(a,S)$, of the form

\begin{equation}
\Op(a,S)\psi(x)=\int\psi(\bq)a(x,\bp)e^{\frac{i}{\hbar}S(\bp,\bq,x)}\frac{d\bp d\bq}{(2\pi\hbar)^{d}}\label{f8}
\end{equation}

from the space $C_{0}^{\infty}(\V)$ of compactly supported smooth
functions on $\R^{d}$ to the space $\mathcal{D}'(\V)$ of distribution
on $\V$, where 

\begin{itemize}

\item $\hbar$ is a fixed real number in the interval $[0,1]$ (later on,
we will be interested in taking the limit $\hbar\rightarrow0$ and
in considering $\hbar$ as a formal parameter in the resulting asymptotic
expansion), 

\item $a$ is a smooth function on $\V\times(\V)^{*}$ called the \textbf{amplitude}
or the \textbf{(total) symbol} of the Fourier integral operator, 

\item $S$ is a smooth function on $(\V)^{*}\times\V\times\V$ called the
\textbf{phase} of the operator. (More generally, one can define the
phase on $\Lambda\times\V\times\V$, where $\Lambda$ is a more general
space of parameters than $(\V)^{*}$; see \cite{D,Ho} for a presentation
of the full theory.)

\end{itemize}

A general problem is to find suitable conditions on both the amplitudes
and the phases so as to obtain a class of FIOs that enjoys the following
nice properties:
\begin{itemize}

\item the operator composition is closed when restricted to this class of
FIOs (which is in general not the case)

\item the operators can be extended to continuous operators on the space
$L^{2}(\V)$ of square integrable functions on $\V$

\end{itemize}
We now present two classes of FIOs that have these good properties.

\subsection{Pseudodifferential operators}

A \textbf{pseudodifferential operator} is a Fourier integral operator
with phase $S(\bar{\xi},\bar{x},x)=\langle\bar{\xi},x-\bar{x}\rangle$.
In other words, it is a integral operator of the form 

\begin{equation}
\big(\Op(a)\psi\big)(x)=\int\psi(\bq)a(x,\bp)e^{{\frac{i}{\hbar}\langle\bp,(x-\bq)\rangle}}\m.\label{f4}
\end{equation}

Following \cite[p. 12]{M}, we define $S_{n}(1)$ to be the set of
\textbf{bounded symbols} (or amplitudes) on $\R^{n}$, that is, the
set of families of smooth functions on $\R^{n}$ parametrized by some
$\hbar\in(0,\hbar_{0}]$ that are uniformly bounded together with
all their derivatives.

Unless necessary, we will not write explicitly the dependence on $\hbar$
(i.e. we will write $a(z)$ instead of $a(z;\hbar)$ for symbols in
$S_{n}(1)$, where $z\in\R^{n}$). 

We will make use of the following result, which is a weaker version
of \cite[Thm. 2.8.1, p. 43]{M}:

\begin{thm}\label{thm:CV} 
If $a\in S_{2d}(1)$, then $\Op(a)$ is a continuous operator on $L^{2}(\V)$. 
\end{thm}

The class of pseudodifferential operators with bounded symbols is
closed under composition: Namely, we have that
\[
\Op(a)\circ\Op(b)=\Op(a\spm b),
\]
where
\begin{equation}
(a\spm b)(x,\xi)=\int a(x,\bp)b(\bq,\xi)e^{\frac{i}{\hbar}\langle(\bp-\xi),(x-\bq)\rangle}\m.\label{f5}
\end{equation}
is the \textbf{Standard product} between (bounded) symbols (see \cite{EZ}
for instance). 

\subsection{A class of bounded $\fio$s}

Let $\mathrm{Diff}(\V)$ be the group of diffeomorphisms of $\V$.
We will now consider Fourier integral operators $\Op(a,S)$ for which
the phase is of the form
\[
S(\bar{\xi},\bar{x},x)=\langle\bar{\xi},\varphi^{-1}(x)-\bar{x}\rangle,
\]
where $\varphi$ is a diffeomorphism on $\V$ of a special type. More
precisely, we focus on the cases when $\varphi$ lies in the following
subgroups of the diffeomorphisms on $\R^{d}$:

\begin{defn}
We define 

(1) the subgroup of \textbf{bounded diffeomorphisms} $\DiffB(\R^{d})$
to be the diffeomorphisms of $\R^{d}$ that have all of their derivatives
bounded, i.e. $\sup_{x\in\R^{d}}|\partial_{x}^{\beta}\varphi(x)|<\infty$
for all multi-indices $\beta\in\mathbb{N}^{d}\backslash(0,\dots,0)$;

(2) the subgroup of \textbf{volume preserving diffeomorphisms} $\DiffBV(\R^{d})$
to be the subgroup of $\DiffB(\R^{d})$ such that $\vert\varphi'(x)\vert=1$.
\end{defn}

\begin{prop} \label{prop:continuity}
Given $a\in S_{2d}(1)$ and $\varphi\in\DiffB(\V)$, the corresponding FIO 
\begin{equation}
\Op(a,\varphi)\psi(x)=\int\psi(\overline{x})a(x,\overline{\xi})e^{\frac{i}{\hbar}\langle\overline{\xi},\varphi^{-1}(x)-\overline{x}\rangle}\frac{d\overline{\xi}d\overline{x}}{(2\pi\hbar)^{d}},\label{f71}
\end{equation}
is a continuous linear operator on $L^{2}(\V)$. Moreover, if $\varphi\in\DiffBV(\R^{d})$
and the amplitude satisfies the additional condition
\begin{equation}
\frac{1}{(2\pi\hbar)^{\frac{d}{2}}}\int a^{*}(\varphi(x),\xi)a(\varphi(x),\bar{\xi})e^{\frac{i}{\hbar}\langle x,\xi-\bar{\xi}\rangle}dx=\delta(\xi-\bar{\xi}),\label{eq:adjoint}
\end{equation}
where $\delta$ is the delta function, then $\Op(a,\varphi)$ is unitary. 
\end{prop}

\begin{proof}
First consider the action of $\Diff(\V)$ on $C^{\infty}(T^{*}\V)$
defined by 

\begin{equation}
(\varphi a)(x,\xi)=a(\varphi^{-1}x,\xi),\,\varphi\in\Diff(\R^{d})\label{f70}
\end{equation}

and the action by pullback of $\Diff(\V)$ on the functions on $\R^{d}$:
\[
t_{\varphi}(\psi)(x)\,:=\,\psi(\varphi^{-1}(x)).
\]
These two actions have the following properties:

\begin{itemize}
\item the action (\ref{f70}) restricted to the subgroup $\DiffB(\R^{d})$
preserves the space of bounded symbols $S_{2d}(1)$;
\item if $\varphi\in\DiffB(\V)$, $t_{\varphi}$ is a continuous operator
on $L^{2}(\V)$ and, if $\varphi\in\DiffBV(\R^{d})$, it is also unitary.
\end{itemize}
The proof of the Proposition will follow now from the follwing identity 
\[
\Op(a,\varphi)=t_{\varphi}\circ\mathrm{Op}(\varphi^{-1}a),
\]
together with Theorem \ref{thm:CV} and from the fact that, for $\varphi\in\DiffB(\R^{d})$
(resp. $\varphi\in\DiffBV(\R^{d})$), $t_{\varphi}$ is continuous
(resp. unitary) while $\varphi^{-1}a$ remains bounded when $a$ is
bounded. A direct computation shows that (\ref{eq:adjoint}) implies
unitarity. 
\end{proof}

We now show that the FIOs of the form (\ref{f71}) are closed under
operator composition by defining a product-like operation for their
symbols. (Note that we will not obtain an algebra of symbols here,
since this new symbol product depends on the particular underlying
diffeomorphisms.)

\begin{prop}
\noindent \label{pro01} Let $\varphi_{1},\varphi_{2}\in\mathrm{Diff}_{b}(\V)$
and $a,b\in S_{2d}(1)$, then 
\begin{equation}
\Op(a,\varphi_{1})\circ\Op(b,\varphi_{2})=\Op(a{_{\varphi_{1}}\star_{\varphi_{2}}}b,\,\varphi_{2}\circ\,\varphi_{1})\label{eq:composition}
\end{equation}
where $(a{_{\varphi_{1}}\star_{\varphi_{2}}}b)(x,\xi)$ is given by
the integral 
\begin{equation}
\int a(x,\bp)b(\bq,\xi)e^{\frac{i}{\hbar}\big(\langle\xi,\varphi_{1}^{-1}(\bq)-\varphi_{1}^{-1}\circ\varphi_{2}^{-1}(x)\rangle+\langle\bp,\varphi_{1}^{-1}(x)-\bq\rangle\big)}\m\label{f111}
\end{equation}
\end{prop}
 
\begin{proof}
\noindent We first compute the composition
\[
\Big(\Op(a,\varphi_{1})\Op(b,\varphi_{2})\psi)\Big)(x)
\]
 directly, and we obtain 
\[
\int\psi(\tq)a(x,\bp)b(\bq,\tp)e^{\frac{i}{\hbar}\big(\langle\bp,\varphi_{1}^{-1}(x)-\bq\rangle+\langle\tp,\varphi_{2}^{-1}(\bq)-\tq\rangle\big)}\m\n.
\]
The phase of the oscillatory exponential in the line above can be
rewritten as follows
\[
\langle\tp,(\varphi_{2}\circ\varphi_{1})^{-1}(x)-\tq\rangle+\Big(\langle\tp,\varphi_{2}^{-1}(\bq)-\varphi_{1}^{-1}\circ\varphi_{2}^{-1}(x)\rangle+\langle\bp,\varphi_{1}^{-1}(x)-\bq\rangle\Big)
\]
so that, defining the product $(a{_{\varphi_{1}}\star_{\varphi_{2}}}b)$
as in (\ref{f111}), we obtain (\ref{eq:composition}).
\end{proof}

\begin{lem}
Let $\varphi_{1},\varphi_{2},\varphi_{3}\in\mathrm{Diff}_{b}(\V)$
and $a,b,c\in S_{2d}(1).$ Then
\begin{equation}
a_{\,\varphi_{1}}\star_{\varphi_{2}\varphi_{3}}(b_{\,\varphi_{2}}\star_{\varphi_{3}}c)=(a_{\,\varphi_{1}}\star_{\varphi_{2}}b)_{\,\varphi_{1}\varphi_{2}}\star_{\varphi_{3}}c.\label{eq:amplitude_composition}
\end{equation}
\end{lem}

\begin{proof}
This comes immediately from the fact that 
\[
\Op(a,\varphi_{1})\circ\Big(\Op(b,\varphi_{2})\circ\Op(c,\varphi_{3})\Big)=\Big(\Op(a,\varphi_{1})\circ\Op(b,\varphi_{2})\Big)\circ\Op(c,\varphi_{3})
\]
together with (\ref{eq:composition}).
\end{proof}

\subsection{Asymptotic expansions and formal operators\label{sub:Asymptotic-expansion}}

We now work out the asymptotic expansion of the bounded operators
(\ref{f71}) in the limit $\hbar\rightarrow0$. First, we fix the
dependence in $\hbar$ for the amplitude as follows
\begin{equation}
a(x,\xi)=a^{0}(x,\xi)+a^{1}(x,\xi)\hbar+a^{2}(x,\xi)\hbar^{2}+\cdots,\label{eq:asymptotic}
\end{equation}
where the $a^{n}\in S_{2d}(1)$ do not depend on $\hbar$ for all
$n$. Namely, the Borel summation lemma (see \cite[Prop. 2.3.2, p. 14]{M}
for instance) guarantees then that there exists an amplitude in $S_{2d}(1)$
depending on $\hbar$ whose asymptotic expansion in $\hbar$ yields
back (\ref{eq:asymptotic}). 

Now, changing the variable $\tilde{\xi}=\xi/\hbar$ and letting $\hbar\rightarrow0$
(which allows us to perform a Taylor's series of the amplitude at
$(x,0)$), we obtain that 
\begin{eqnarray*}
\Op(a(x,\xi),\varphi)\psi(x) & = & \int\psi(\overline{x})a(x,\hbar\tilde{\xi})e^{i\langle\tilde{\xi},\varphi^{-1}(x)-\overline{x}\rangle}\frac{d\tilde{\xi}d\overline{x}}{(2\pi)^{d}},\\
 & = & \sum_{n\geq0}\hbar^{n}\op(P^{n},\varphi)\psi(x).
\end{eqnarray*}
where $\op$ is the same integral operator as $\Op$ except with the
parameter $\hbar$ in the phase set to $1$, and where
\[
P^{n}(x,\xi)=\sum_{k=0}^{n}f_{\alpha}(x)\xi^{\alpha},
\]
are polynomial in $\xi$ of order $n$ with coefficients in $S_{d}(1)$
(actually, $f_{\alpha}(x)=\frac{1}{|\alpha|!}\partial_{\xi}^{\alpha}a_{n-|\alpha|}(x,0)$). 

Since, for a polynomial $P^{n}(x,\xi)$ in $\xi$ as above, the corresponding
operator 
\[
\op(P^{n},\varphi)\psi(x)=\sum_{|\alpha|\leq n}f_{\alpha}(x)(D_{x}^{\alpha}\psi)(\varphi^{-1}(x))=\left(P^{n}\left(x,D\right)\psi\right)(\varphi^{-1}(x))
\]
is a differential operator of order $n$ (composed with a pullback),
we $\mathcal{}$obtain for $\Op(a,\varphi)$ an asymptotic expansion
in terms of infinite order differential operators of the form: 
\begin{equation}
\Op(a,\varphi)\psi(x)=P^{0}(x)\psi(\varphi^{-1}(x))+\sum_{n\geq1}\hbar^{n}\left(P^{n}\left(x,D\right)\psi\right)(\varphi^{-1}(x)).\label{eq: formal operators}
\end{equation}

\begin{rem}
This derivation for the asymptotic (\ref{eq: formal operators}) is
a shortcut for the usual stationary phase expansion. One recovers
(\ref{eq: formal operators}) by using the usual stationary phase
expansion (see \cite{EZ}) for quadratic phase using the following
change of variable $\bar{y}=\varphi^{-1}(x)-\bar{x}$.
\end{rem}

In the following definition, we retain only the formal aspects of
the asymptotics, forgetting that the operators (\ref{f71}) are actually
bounded operators (i.e. the amplitudes are in $S_{2d}(1)$ and the
action is in $\mathrm{Diff}_{b}(\V)$). This will allows us later
on to consider quantizations of actions that are not necessarily volume-preserving
nor bounded. 

\begin{defn}
We define the algebra $\mathcal{D}$ of formal operators of the form
\begin{equation}
\Op_{1}(P,\varphi)\psi(x)=P^{0}(x)\psi(\varphi^{-1}(x))+\sum_{n\geq1}\hbar^{n}\left(P^{n}\left(x,D\right)\psi\right)(\varphi^{-1}(x)),\quad\varphi\in\Diff(\V)\label{eq:formal operators}
\end{equation}
which acts on the formal space of functions $C^{\infty}(\R^{d})[[\hbar]]$,
and where 
\[
P^{n}(x,D)=\sum_{|\alpha|\leq n}f_{\alpha}(x)D^{\alpha},
\]
is a differential operator of order $n$ with coefficients $f_{\alpha}\in C^{\infty}(\R^{d})$.
The corresponding space of symbols $\mathcal{P}$ is the space of
formal functions of the form
\[
P(x,\xi)=P^{0}(x)+\sum_{n\geq1}\hbar^{n}\sum_{|\alpha|\leq n}f_{\alpha}(x)\xi^{\alpha},
\]
with $P^{0}(x),f_{\alpha}\in C^{\infty}(\R^{d})$. 
\end{defn}

Note that, as before, we obtain a composition of formal symbols of
thanks to 
\begin{equation}
\op(P,\varphi_{1})\circ\op(K,\varphi_{2})=\op(P\,_{\varphi_{1}}\star_{\varphi_{2}}K,\varphi_{1}\circ\varphi_{2}).\label{eq:formal product}
\end{equation}
Again, this does not define an algebra structure on $\mathcal{P}$
since the composition depends on the underlying bounded diffeomorphisms
$\varphi_{1}$ and $\varphi_{2}$.

\section{Quantization of $G$-actions\label{sec:Quantization-of--actions}}

In this section, we \textit{quantize} a given action $\varphi$ of
a group $G$ on $\R^{d}$ (using the Fourier integral operators of
the previous section as well as their asymptotics). By this, we mean
to associate with $\varphi$ a set $\Rep_{\varphi}(G)$ of infinite
dimensional representations of $G$ on an appropriate space of {}``functions''
on $\R^{d}$. We call \textbf{a quantization} \textbf{of the action}
a representation in $\Rep_{\varphi}(G)$. 

The actual implementation of $\Rep_{\varphi}(G)$ (i.e. the choice
of the functional space on which we represent the group as well as
the properties of the operators forming the quantization) depends
on the type of actions at hand. We distinguish between three cases,
all of which contain what we call the \textbf{trivial quantizatio}n,
i.e. the representation obtained by pullback of functions:
\begin{equation}
(T_{g}\psi)(x)=\psi(\varphi_{g^{-1}}x),\quad g\in G.\label{f180-1}
\end{equation}
Quantizations in $\Rep_{\varphi}(G)$ can be regarded, in a sense,
as {}``deformations'' of the trivial quantization. 

Here are the three cases we are interested in:
\begin{itemize}

\item An action of a group $G$ on $\R^{d}$ by \textsl{smooth} diffeomorphisms
(i.e. $\varphi_{g}\in\Diff(\R^{d})$ for all $g\in G$), which we
call here simply an action. 

\item An action of a group on $G$ on $\R^{d}$ by \textsl{bounded smooth}
diffeomorphisms (i.e. $\varphi_{g}\in\DiffB(\R^{d})$ for all $g\in G$),
which we call a \textbf{bounded action}.

\item An action of a group $G$ on $\R^{d}$ by \textsl{bounded and volume-preserving
smooth} diffeomorphisms (i.e. $\varphi_{g}\in\DiffBV(\R^{d})$ for
all $g\in G$), which we call a \textbf{volume-preserving action}. 

\end{itemize}

\subsection{Unitary G-systems}

If the action is a \textit{volume-preserving and bounded} (i.e $\varphi_{g}\in\DiffBV(\V)$
for all $g\in G$), $\Rep_{\varphi}(G)$ is a set of representations
by \textit{bounded} \textit{unitary} operators on the Hilbert space
$L^{2}(\V)$ of the square integrable functions on $\V$. These operators
are of the form (\ref{f71}), with amplitudes in $S_{2d}(1)$ and
satisfying the unitarity condition (\ref{eq:adjoint}). Observe that
the trivial quantization in this case is formed by bounded unitary
operators on $L^{2}(\R^{d})$. We define:

\begin{defn}
\label{d4} A \textbf{unitary $G$-system of amplitudes} (associated
with a volume-preserving and bounded action $\varphi$ of a group
$G$ on $\V$) is a map
\begin{equation}
a:G\longrightarrow S_{2d}(1)
\end{equation}
such that the collection of operators
\[
T_{g}^{a}:=\Op(a_{g},\varphi_{g})
\]
(defined in (\ref{f71})) forms a unitary representation of $G$ by
bounded operators on $L^{2}(\V)$. 
\end{defn}

\begin{example}
\textbf{Unitary G-systems from geometric quantization}.\label{sub:Geometric-quantization-intermezz}
Suppose we have a volume-preserving action $\varphi$ of a Lie group
$G$ on $\bR^{2n}$ (endowed with its canonical symplectic form $\omega=\sum_{i}dp_{i}\wedge dx^{i}$),
which is hamiltonian and which admits a momentum map
\[
J:\mathfrak{g}\rightarrow C^{\infty}(\bR^{2n}).
\]
The condition $|\varphi'(x)|=1$ is always satisfied, since $\varphi_{g}:\bR^{2n}\rightarrow\bR^{2n}$
is a symplectomorphism for all $g\in G$; so, here, the \textit{volume-preserving}
condition on $\varphi$ is only really a condition on the boundedness
of $\varphi_{g}$ as well as on its derivatives. 

Geometric quantization prescribes then a way (as explained in \cite[ch. 8. sec. 4]{W}
for instance) to associate a unitary flow on $L^{2}(\bR^{2n})$ with
the hamiltonian flow $\varphi_{t}$ integrating the hamiltonian vector
field $X_{f}$ of a function $f\in C^{\infty}(\R^{2n})$; namely,
\[
U_{t}(f)\Psi(x)=\exp\left(i\int_{0}^{t}\mathcal{L}_{f}(\varphi_{s}^{-1}(x))ds\right)\Psi(\varphi_{t}^{-1}(x)),
\]
where $\mathcal{L}_{f}=\theta(X_{f})-f$, $\theta$ is the canonical
Liouville $1$-form on $\R^{2n}$ and $x=(p,q)\in\bR^{2n}$. If $f$
is a complete function (i.e. $X_{f}$ is a complete vector field),
then $U_{t}(f)$ forms a $1$-parameter group. 

Now, if the Lie group $G$ is nilpotent for instance and the hamiltonian
vector fields $X_{J(v)}$ are complete for all $v\in\mathfrak{g}$,
we obtain a unitary representation of $G$ on $L^{2}(\bR^{2n})$ by
taking
\[
\rho_{g}:=U_{1}(J(v)),\quad g=\exp(v),
\]
where $\exp$ is the exponential map from $\mathfrak{g}$ to $G$,
which is a diffeomorphism for nilpotent groups. 

Observe that, if we set 
\[
a_{g}(x)=\exp\left(i\int_{0}^{1}\mathcal{L}_{J(v)}(\varphi_{s}^{-1}(x))ds\right),\quad g=\exp(v),
\]
the representation $\rho_{g}$ can be regarded as a quantization $T_{g}^{a}\in\Rep_{\varphi}(G)$
associated with the $G$-system $a_{g}$, which is independent of
$\xi$. 
\end{example}

\begin{example}
\textbf{\label{Unitary-G-systems-from-Galilean-boost}Unitary G-systems
from galilean covariance.} Consider the space-time $\R^{4}=\R_{t}\times\R_{x}^{3}$.
The additive group $\R_{v}^{3}$ translations acts on $\R^{4}$ by
galilean boost 
\[
\varphi_{v}(t,x)=(t,x+vt).
\]
In (non-relativistic) quantum mechanics, dynamics is described by
square integrable functions $\Psi:\R^{4}\rightarrow\mathbb{C}$ satisfying
the Schr\"odinger equation $i\partial_{t}\Psi=H\Psi$, where $H$
is the Hamiltonian operator. It turns out that this equation is not
covariant with respect to the trivial quantization of the galilean
boost. To obtain covariance, one needs to use the following unitary
$G$-system $a_{v}(t,x,\xi)=e^{-i(\frac{1}{2}mv^{2}t-mvx)}$, which
yields the quantization
\[
T_{v}^{a}\Psi(t,x)=e^{-i(\frac{1}{2}mv^{2}t-mvx)}\psi(t,x-vt).
\]
\end{example}

There seems to be many examples in the literature of unitary $G$-systems
that are independent of $\xi$ as in the previous examples (see also
for instance the representation in \cite[p. 544, p. 557]{Rieffel2}
)Because of this, we devote Section \ref{sub:-systems-independent-of-xi}
to the study of these special $G$-systems. 

Let us give an example of unitary $G$-system that also depends on
$\xi$. 

\begin{example}
Consider the multiplicative group $\R^{+}$ of the strictly positive
real numbers and its trivial action on $\R$, i.e. $\varphi_{g}(x)=x$.
Then 
\[
a_{g}(\xi)=e^{\frac{i}{\hbar}\xi\ln(g)}
\]
is a unitary $G$-system. Namely, one verifies that the corresponding
operator is then given by
\[
T_{g}\psi(x)=\psi(x+\ln(g)),
\]
which is a unitary representation on $L^{2}(\R)$. 
\end{example}

If the action is only \textit{bounded}, (i.e. $\varphi_{g}\in\DiffB(\V)$
for all $g\in G$), $\Rep_{\varphi}(G)$ is again a set of infinite
dimensional representations on $L^{2}(\R^{d})$, except that now the
operators (\ref{f71}) forming the representations are no longer unitary.
Condition (\ref{eq:adjoint}) on the amplitudes is then dropped, but
we still require that the amplitudes are in $S_{2d}(1)$. In this
case, the operators $T_{g}:L^{2}(\R^{d})\rightarrow L^{2}(\R^{d})$
forming the trivial quantization are only bounded but non longer unitary,
since the action is not volume-preserving. We define:

\begin{defn}
\label{d4-1} A \textbf{(non-unitary) $G$-system of amplitudes} (associated
with a bounded action $\varphi$ of a group $G$ on $\V$) is a map
\begin{equation}
a:G\longrightarrow S_{2d}(1)
\end{equation}
such that the collection of operators
\[
T_{g}^{a}:=\Op(a_{g},\varphi_{g})
\]
(defined in (\ref{f71})) forms a representation of $G$ by bounded
operators on $L^{2}(\V)$. 
\end{defn}

\subsection{Formal G-systems\label{sub:Formal-G-systems}}

The operators in (\ref{f71}) that we used to define quantizations
of actions in the two previous cases depend on a parameter $\hbar$
and, thus, have an asymptotic expansion in terms of formal operators
as discussed in Section \ref{sub:Asymptotic-expansion}. 

We can now forget that these expansions comes from well-defined bounded
operators on $L^{2}(\R^{d})$ and use the formal operators (\ref{eq:formal operators})
to define formal quantizations when the action is neither bounded
nor volume-preserving. In this case, $\Rep_{\varphi}(G)$ is a set
of formal representations by formal operators of the form (\ref{eq:formal operators})
on the space $C^{\infty}(\R^{d})[[\hbar]]$ of formal power series
in the formal parameter $\hbar$ with value in the smooth function
on $\R^{d}$. More precisely, we define:
\begin{defn}
A \textbf{formal $G$-system of amplitudes} (associated with an action
of a group $G$ on $\R^{d}$) is a map
\[
a:G\longrightarrow\mathcal{P}
\]
such that the formal operators
\[
T_{g}^{a}=\op(a_g,\varphi_{g})\in\mathcal{D}
\]
form a representation of $G$ on $C^{\infty}(\R^{d})[[\hbar]]$, where
$\op(a,\varphi_{g})$ is defined as in (\ref{eq:formal operators}).
\end{defn}
Formal $G$-systems seem to be related to both deformation quantization
($G$-equivariant star-products) and deformation theory of Lie morphisms,
when the action we start with is a smooth action of a \textit{Lie
group} on $\R^{d}$. Let us comment here briefly on these points. 

In deformation quantization (\cite{BFFLS}), one quantizes an action
(by Poisson diffeomorphisms) of a Lie group $G$ on a Poisson manifold
$M$ by constructing $G$-equivariant star-products $\star$ on $M$.
For us, $M=\R^{d}$ . (This notion is somewhat different whether one
considers formal deformations, as in \cite{BFFLS}, or strict ones,
as in \cite{Rieffel2}.) The idea is to find star-products $\star$
(quantizing a Poisson structure on $\R^{d}$ that is invariant with
respect to the group action), which has the following property ($G$-equivariance):
\[
T_{g}\psi_{1}\star T_{g}\psi_{2}=T_{g}(\psi_{1}\star\psi_{2}),
\]
where $\psi_{1},\psi_{2}\in C^{\infty}(\R^{d})[[\hbar]]$ and $T$
is the aforementioned trivial quantization of the action. 

Despite compatibility between the action and the Poisson structure,
the star-products quantizing the Poisson structure are generally not
$G$-equivariant ($G$-equivariant star-products may even not exist
at all; see \cite{As}). Thus, in some cases, one also needs to {}``deform''
the trivial quantization to obtain $G$-equivariance (for the deformed
action), as in \cite{A1,A2} in the formal case (for the corresponding
infinitesimal action), or as in \cite{Rieffel2} for strict quantization
of the Heisenberg manifolds. 

The latter case is specially interesting for us, since the deformation
of the action $\varphi$ of the Heisenberg group $G$ on the Heisenberg
manifolds is of the form (\ref{eq:deformation}) for a certain $G$-system
independent of $\xi$ (see \cite[p. 557]{Rieffel2}). 

It would be interesting to see if, for a given (strict) star-product
on a Poisson manifold on which a Lie group $G$ acts by Poisson diffeomorphisms,
one can always find a deformation of the trivial quantization in our
space of quantization $\Rep_{\varphi}(G)$ that is $G$-equivariant. 

There is also a way in which quantizations of actions by $G$-systems
as defined above may be related to the general theory of Lie morphism
deformations as in \cite{NiRi}, and, more specifically, to the work
of Ovsienko and collaborators (\cite{AgAmLeOv,Ov}) on embeddings
of the Lie algebra of vector fields into various Lie algebras (and,
in particular, the Lie algebra of pseudodifferential operators). 

Namely, the infinitesimal version of the trivial quantization of an
action yields an embedding from a Lie subalgebra of the vector fields
on the manifold into its Lie algebra of pseudodifferential operators.
Then the infinitesimal representations associated with quantizations
in $\Rep_{\varphi}(G)$ (i.e., the space of Lie algebra representations
corresponding to the unitary/formal representations in $\Rep_{\varphi}(G)$)
should, in a sense, be related to deformations of this embedding. 

It would also be interesting to compare the various obstructions (and
actual deformations) obtained in this infinitesimal context with the
obstructions we obtain in Section \ref{sec:Existence-and-rigidity}.

\subsection{Unitary $G$-systems independent of $\xi$\label{sub:-systems-independent-of-xi}}

Let $\varphi$ be a volume-preserving action of $G$ on $\V$. We
are looking for $G$-systems associated with this action for which
the amplitudes do not depend on $\xi$. Example \ref{sub:Geometric-quantization-intermezz}
from geometric quantization and Example \ref{Unitary-G-systems-from-Galilean-boost}
from the galilean covariance of the Schr\"odinger equation are of
this type. In the context of strict deformation quantization the representations
in \cite[p. 544, p. 557]{Rieffel2} are also of this type. 

Let us study these $G$-systems independently. We start by defining
a useful complex:

Denote by $\mathcal{B}$ the space of smooth functions on $\V$ with
all of their derivatives bounded. One verifies that $\mathcal{B}$
is a left $G$-module with respect to the action
\[
(g\cdot S)(x)=S(\varphi_{g}^{-1}(x)),\quad S\in\mathcal{B}.
\]
Observe that, in contrast with $S_{d}(1)$, we do not require that
a function in $\mathcal{B}$ be bounded (only its derivatives). We
further turn $\mathcal{B}$ into a $G$-bimodule by considering the
right action of $G$ on $\mathcal{B}$. Now consider the group cohomology
with values in the bimodule $\mathcal{B}$. The corresponding space
$C_{\varphi}^{k}(G,\mathcal{B})$ of (normalized) $k$-cochains is
given by the smooth maps
\[
S:G^{k}\longrightarrow\mathcal{B},\quad k\geq0
\]
such that $S_{g_{1},\dots,g_{k}}=0$ if one of the $g_{i}$'s is the
group unit. The differential
\[
\delta:C_{\varphi}^{k}(G,\mathcal{B})\rightarrow C_{\varphi}^{k+1}(G,\mathcal{B})
\]
is given by the usual formula
\[
(\delta c)_{g_{1},\dots,g_{k+1}}=g_{1}\cdot c_{g_{2},\dots,g_{k+1}}-c_{g_{1}g_{2},\dots,g_{k}}+\cdots\pm c_{g_{1},\dots,g_{k}g_{k+1}}\mp c_{g_{1},\dots,g_{k}},
\]
where the right action by $g_{k+1}$ on the last term is the trivial
action. 

\begin{thm}
\label{l27}A $G$-system is independent of $\xi$ iff it is of the
form
\begin{equation}
a_{g}(x)=e^{iS_{g}(x)},\label{eq27}
\end{equation}
where $S_{g}$ is a $1$-cocycle in $C_{\varphi}^{\bullet}(G,\mathcal{B})$.
The corresponding operators are given by
\begin{equation}
T_{g}^{a}\psi(x)=e^{iS_{g}(x)}\psi(x).\label{eq:def_by_phase}
\end{equation}
Moreover, cocycles in the same cohomology class induce equivalent
representations. In other words, $H_{\varphi}^{1}(G,\mathcal{B})$
controls the deformations by unitary multiplication operators of the
trivial quantization: If this first cohomology group vanishes, all
deformations of the form (\ref{eq:def_by_phase}) are equivalent to
the trivial quantization.
\end{thm}

\begin{proof}
Suppose $a_{g}$ is of the form (\ref{eq27}). Since $S_{g}\in\mathcal{B}$,
we have that $a_{g}\in S_{2d}(1)$, and Proposition \ref{prop:continuity}
guarantees that $T_{g}^{a}$ is a continuous operator on $L^{2}(\V)$.
The unitarity follows from the fact that $a_{g}^{*}(x)a_{g}(x)=1$
for all $x\in\V$. Conversely, the operators corresponding to a $G$-system
$a_{g}(x)$ that is independent of $\xi$ are of the form
\[
T_{g}^{a}\psi(x)=a_{g}(x)\psi(\varphi_{g}^{-1}(x)).
\]
The unitarity condition for these operators is equivalent to the condition
\[
\int\big(1-a_{g}^{*}(\varphi_{g}(x))a_{g}(\varphi_{g}(x))\big)\psi_{1}^{*}(x)\psi_{2}(x)=0,
\]
for all $\psi_{1},\psi_{2}\in L^{2}(\V)$, which in turns is equivalent
to $a_{g}^{*}(x)a_{g}(x)=1$. The only functions satisfying this last
condition are of the form $e^{iS_{g}(x)}$. Now observe that, for
an amplitude of this form, $a_{g}\in S_{2d}(1)$ if and only if $S_{g}\in\mathcal{B}$. 

Let us check now that $a_{g}(x)=e^{iS_{g}(x)}$ with $S\in C^{1}(G,\mathcal{B})$
is a $G$-system if and only if $\delta S=0$. For this, we observe
that
\[
T_{g_{1}g_{2}}^{a}\psi(x)=e^{iS_{g_{1}g_{2}}(x)}\psi(\varphi_{g_{1}g_{2}}^{-1}(x))
\]
is equal to
\[
T_{g_{1}}^{a}T_{g_{2}}^{a}\psi(x)=e^{i(S_{g_{1}}(x)+S_{g_{2}}(\varphi_{g_{1}}^{-1}(x)))}\psi(\varphi_{g_{1}g_{2}}^{-1}(x))
\]
if and only if
\[
S_{g_{2}}(\varphi_{g_{1}}^{-1}(x))-S_{g_{1}g_{2}}(x)+S_{g_{1}}(x)=0,
\]
that is if and only if $\delta S=0$. At last, let us notice that
the normalization condition for cochains $a\in C^{1}(G,\mathcal{B})$
is equivalent to $T_{e}^{a}=\operatorname{id}$. 

Let us show now that if $S-\tilde{S}=\delta K$, where $S$ and $\tilde{S}$
are $1$-cocycle and $K$ is a $0$-cochain, then the induced representations
$T^{a}$ and $T^{\tilde{a}}$ are equivalent. Consider the bounded
operator $\hat{K}\psi(x)=e^{iK(x)}\psi(x)$. Then
\begin{eqnarray*}
(T_{g}^{a}\circ\hat{K})\psi(x) & = & e^{i(S_{g}(x)+K(\varphi_{g}^{-1}(x)))}\psi(\varphi_{g}^{-1}(x)),\\
(\hat{K}\circ T_{g}^{\tilde{a}})\psi(x) & = & e^{i(\tilde{S}_{g}(x)+K(x))}\psi(\varphi_{g}^{-1}(x)).
\end{eqnarray*}
Therefore, the relation $\tilde{S}_{g}(x)-S_{g}(x)=K(\varphi_{g}^{-1}(x))-K(x)=(\delta K)_{g}(x)$
implies that $T_{g}^{a}\circ\hat{K}=\hat{K}\circ T_{g}^{\tilde{a}}$. 
\end{proof}

In the next section, we will work out a similar cohomological equation
(a Maurer-Cartan equation) for general $G$-systems (i.e. with a dependence
on $\xi$). 

\begin{example}
\label{ex29}Let $h\in\mathcal{B}$ be invariant under the action
of $G$ (i.e. $h(\varphi_{g}^{-1}(x))=h(x)$ for all $x$ and $g$).
For any smooth function $c:G\rightarrow\bR$ that satisfies
\begin{equation}
c_{1}=0\quad\textrm{ and }\quad c_{g_{1}g_{2}}=c_{g_{1}}+c_{g_{2}},\label{eq:29}
\end{equation}
we verify that $S_{g}(x)=h(x)c_{g}$ is a cocycle. As a consequence,
the family of amplitudes
\[
a_{g}(x)=e^{i\sum_{k}h_{k}(x)c_{g}^{k}},\quad g\in G,
\]
is a $G$-system, where $h_{1},\dots,h_{n}$ are invariant functions
in $\mathcal{B}$ and $c^{1},\dots,c^{n}$ are smooth functions from
$G$ to $\bR$ satisfying (\ref{eq:29}). 
\end{example}


\section{DGAs of $G$-amplitudes\label{sec:The-complex-of}}

In this section, we construct two DGAs, $\mathcal{A}_{\varphi}$ and
$\mathcal{P}_{\varphi}$, associated with, respectively, a bounded
action and a smooth action $\varphi$ of a group $G$ on $\R^{d}$.
We show that the Maurer-Cartan elements in $\mathcal{A}_{\varphi}$
correspond to $G$-systems while Maurer-Cartan elements in $\mathcal{P}_{\varphi}$
correspond to formal $G$-systems. One can regards $\mathcal{P}_{\varphi}$
as the {}``asymptotic'' version of $\mathcal{A}_{\varphi}$. We
also show that, in both cases, gauge equivalent Maurer-Cartan elements
yields equivalent quantizations.

\subsection{MC elements in $\mathcal{A}_{\varphi}$ and $G$-systems}

We define here a Differential Graded Algebra (or DGA for short) associated
with a bounded action $\varphi$ of a Lie group $G$ on $\V$ whose
Maurer-Cartan elements correspond to (nonunitary) $G$-systems of
amplitudes. Roughly, the elements of degree $k$ in this DGA are amplitudes
depending on $k$ group variables, and the graded product corresponds
to the composition of the Fourier integral operators that one can
naturally associate with these amplitudes using the action as a phase. 

More precisely, for any $k\geq0$, we define the space of $k$-cochains
by 
\begin{equation}
{\mathcal{A}_{\varphi}}^{k}=\{a:G\times\dots\times G\rightarrow S_{2d}(1)\},\quad\mathcal{A}_{\varphi}^{0}=S_{2d}(1),
\end{equation}
such that $a_{e,\dots,e}=1$ with $e$ being the group unit. The differential
$d:{\mathcal{A}_{\varphi}}^{k}\rightarrow{\mathcal{A}_{\varphi}}^{k+1}$
is defined by
\begin{equation}
(da)(g_{1},\dots,g_{k})=\sum_{i=1}^{k}(-1)^{i}a(g_{1},\dots,g_{i}g_{i+1},\dots,g_{k+1}),\label{f12}
\end{equation}
which we extend by $\bC$-linearity to ${\mathcal{A}_{\varphi}}^{\bullet}=\oplus_{k\geq0}{\mathcal{A}_{\varphi}}^{k}$.
This turns $(\mathcal{A}_{\varphi}^{\bullet},d)$ it into a complex,
which we call the \textbf{complex of $G$-amplitudes}.

Let us now define a graded associative product on $\mathcal{A}_{\varphi}^{\bullet}$.
To an element $a\in\mathcal{A}_{\varphi}^{k}$, we can assign the
following collection of Fourier integral operators
\[
T_{g_{1},\dots,g_{k}}^{a}:=\Op(a_{g_{1},\dots,g_{k}},\varphi_{g_{1}\dots g_{k}}),\quad(g_{1},\dots,g_{k})\in G^{k}.
\]
The composition of these operators for $a\in\mathcal{A}_{\varphi}^{k}$
and $b\in\mathcal{A}_{\varphi}^{l}$ yields

\begin{eqnarray}
T_{g_{1},\dots,g_{k}}^{a}\circ T_{g_{k+1},\dots,g_{k+l}}^{b} 
& = & \Op(a_{g_{1},\dots,g_{k}},\varphi_{g_{1}\dots g_{k}})\circ\Op(b_{g_{k+1},\dots,g_{k+l}},\varphi_{g_{k+1}\dots g_{k+l}})\label{eq:operator_composition}\\
& = & \Op(a_{g_{1},\dots,g_{k}}\tilde{\star}\; b_{g_{k+1},\dots,g_{k+l}},\varphi_{g_{1}\dots g_{k+l}}),\nonumber 
\end{eqnarray}

where $\text{\ensuremath{a_{g_{1},\dots,g_{k}}\tilde{\star}}\;\ \ensuremath{b_{g_{k+1},\dots,g_{k+l}}}}$$ $
is a shorthand for the product of amplitudes defined in (\ref{f111}):
\[
\Big(a(g_{1},\dots,g_{k})\Big)\;_{\varphi_{g_{1}\dots g_{k}}}\star_{\varphi_{g_{k+1}\dots g_{k+l}}}\;\Big(b(g_{k+1},\dots,g_{k+l})\Big).
\]
This leads us to define a graded associative product on the complex
of $G$-amplitudes
\[
\star:\mathcal{A}_{\varphi}^{k}\times\mathcal{A}_{\varphi}^{l}\longrightarrow\mathcal{A}_{\varphi}^{k+l}
\]
in the following way: Given $a\in{\mathcal{A}_{\varphi}}^{k}$ and
$b\in{\mathcal{A}_{\varphi}}^{l}$, we define 
\begin{equation}
(a\star b)(g_{1},\dots,g_{k+l})=\text{\ensuremath{a_{g_{1},\dots,g_{k}}\tilde{\star}}\;\ \ensuremath{b_{g_{k+1},\dots,g_{k+l}}}},\label{f121}
\end{equation}
which turns $\mathcal{A}_{\varphi}^{\bullet}$ into a graded algebra
with the nice property that
\[
T^{a}\circ T^{b}=T^{a\star b}.
\]

\begin{lem}
\label{lem:DGA}$({\mathcal{A}_{\varphi}}^{\bullet},d,\star)$ is
a DGA.
\end{lem}

\begin{proof}
The fact that $d$ squares to zero is clear from its formula (it is
the usual group cohomology differential without the boundary terms).
The associativity of the product $\star$ comes from the associativity
of the operator composition in \ref{eq:operator_composition}

Let us check that $d$ is a derivation for $\star$: 
\begin{eqnarray*}
\delta(a\star b)_{g_{1},\dots,g_{k+l+1}} & = & \sum_{i=1}^{k+l}(-1)^{i}(a\star b)_{g_{1},\dots,g_{i}g_{i+1},\dots,g_{k+l+1}},\\
 & = & \left(\sum_{i=1}^{k}(-i)^{i}a_{g_{1},\dots,g_{i}g_{i+1},\dots g_{k}}\right)\;_{\varphi_{g_{1}\dots g_{k}}}\star_{\varphi_{g_{k+1}\dots g_{k+l}}}b_{g_{k+1},\dots,g_{k+l+1}}\\
 & +(-1)^{k} & a_{g_{1},\dots,g_{k}}\;_{\varphi_{g_{1}\dots g_{k}}}\star_{\varphi_{g_{k+1}\dots g_{k+l}}}\left(\sum_{i=1}^{l}(-1)^{i}b_{g_{k},\dots,g_{k+i-1}g_{k+i},\dots,g_{k+l+1}}\right),\\
 & = & ((\delta a)\star b)_{g_{1},\dots,g_{k+l+1}}+(-1)^{k}(a\star(\delta b))_{g_{1},\dots,g_{k+l+1}}.
\end{eqnarray*}

\end{proof}

Let us now remind the following definition:

\begin{defn}
Let $(A,\star,d)$ be a DGA. The solutions of the \textbf{Maurer-Cartan
equation}
\[
da+a\star a=0
\]
are called \textbf{Maurer-Cartan elements}. The set of all Maurer-Cartan
elements of $A$ will be denoted by $\mathrm{MC}(A)$. 

Two Maurer-Cartan elements $a,b\in\mathrm{MC}(A)$ are called \textbf{gauge
equivalent} if there exists an invertible $u\in A^{0}$ such that
$au-ua=du$. 
\end{defn}

\begin{rem}
The set $\mathrm{MC}(A)$ is a subset of $\mathcal{A}^{1}$. 
\end{rem}

\begin{prop}
\label{prop:MC}Let $\varphi$ be a bounded action of a group $G$
on $\V$. There is a one-to-one correspondence between the set $MC(\mathcal{A}_{\varphi}^{\bullet})$
of Maurer-Cartan elements in the complex of $G$-amplitudes and the
set of (nonunitary) $G$-systems of amplitudes associated with the
action $\varphi$. 
Moreover, gauge equivalent Maurer-Cartan elements induce equivalent
representations.
\end{prop}

\begin{proof}
Let $a\in\mathcal{A}^{1}$. Then the associated collection of operators
$T_{g}^{a}:L^{2}(\V)\rightarrow L^{2}(\V)$ is a representation of
$G$ if and only if
\begin{eqnarray*}
0 & = & T_{g_{1}g_{2}}^{a}-T_{g_{1}}^{a}T_{g_{2}}^{a},\\
0 & = & \Op(a_{g_{1}g_{2}},\varphi_{g_{1}g_{2}})-T_{g_{1}g_{2}}^{a\star a},\\
0 & = & \Op((da)_{g_{1},g_{2}}-(a\star a)_{g_{1},g_{2}},\varphi_{g_{1}g_{2}}),
\end{eqnarray*}
that is iff $da+a\star a=0$. The unitality condition $T_{e}^{a}=\operatorname{id}$
is taken care of by the requirement on the cochains that $a_{e}=1$.

Let us check now that two gauge equivalent Maurer-Cartan elements
induce equivalent representations. First off, we note that, in $\mathcal{A}_{\varphi}^{\bullet}$,
all elements of degree zero are cocycles. This means that $a,b\in\mathrm{MC}(\mathcal{A}_{\varphi}^{\bullet})$
are gauge equivalent if there is an invertible $u\in\mathcal{A}_{\varphi}^{0}$
such that $au=ua$. Since $u$ is of degree zero, neither $u$ nor
$T^{u}$ depend on group variables. The commutation $au=ua$ on the
level of amplitudes implies that
\[
T_{g}^{a}\circ T^{u}=T^{u}\circ T_{g}^{b},\quad g\in G,
\]
on the level of operators. That is, $T^{u}$ intertwines the two representations;
since $T^{u}$ is invertible, because $u$ is invertible, the representations
$T^{a}$ and $T^{b}$ are equivalent. 
\end{proof}

\begin{rem}
The Maurer-Cartan equation applied to an ansatz of the form (\ref{eq27})
yields back the cocycle condition of Proposition \ref{l27}. Namely,
\[
(de^{iS})_{g_{1},g_{2}}(x)=-e^{iS_{g_{1}g_{2}}(x)}\textrm{ and }(e^{iS}\star e^{iS})_{g_{1},g_{2}}(x)=e^{i\big(S_{g_{1}}(x)+S_{g_{2}}(\varphi_{g_{1}}^{-1}(x))\big)},
\]
which implies that $e^{iS}\in\textrm{MC}(\mathcal{A}^{\bullet})$
if and only if $\delta S=0$. 
\end{rem}

\subsection{MC elements in $\mathcal{P}_{\varphi}$ and formal $G$-systems }

We define now a formal version of the amplitude complex by replacing
the bounded symbols $S_{2d}(1)$ by their formal version $\mathcal{P}$. 

The complex of formal $G$-amplitudes $\mathcal{P}_{\varphi}^{\bullet}$
is defined in the following way. The space of $k$-cochain is given
by
\[
{\mathcal{P}_{\varphi}}^{k}=\{a:G\times\dots\times G\rightarrow\mathcal{P}\}.
\]
The differential $d:\mathcal{P}_{\varphi}^{k}\rightarrow\mathcal{P}_{\varphi}^{k+1}$
is obtained from (\ref{f12}) by linear extension. Similarly, we obtained
a graded associative product $\star:\mathcal{P}_{\varphi}^{k}\times\mathcal{P}_{\varphi}^{l}\rightarrow\mathcal{P}_{\varphi}^{k+l}$
from (\ref{f121}) by linear extension, turning $\mathcal{P}_{\varphi}^{\bullet}$
into a DGA (the proof of this is similar to that of Lemma \ref{lem:DGA}).
Mimicking the proof of Proposition \ref{prop:MC}, we obtain: 

\begin{prop}
Let $\varphi$ be an action of $G$ on $\R^{d}$. Then Maurer-Cartan
elements in $\mathcal{P}_{\varphi}^{\bullet}$ are in one-to-one correspondence
with formal $G$-systems associated with $\varphi$. Moreover, gauge
equivalent Maurer-Cartan elements yields equivalent representations.
\end{prop}

\begin{prop}
\label{prop:first_and_second_terms} Let $a=P^{0}(x)+\hbar P^{1}+\cdots\in\mathcal{P}_{\varphi}^{1}$
be a Maurer-Cartan element in $\mathcal{P}_{\varphi}^{\bullet}$.
Then $P^{0}$ is a Maurer-Cartan element in $\mathcal{P}_{\varphi}^{\bullet}$.
It defines a new differential on $\mathcal{P}_{\varphi}^{\bullet}$
as follows:
\begin{equation}
d_{P^{0}}a=da+[P^{0},a]=da+P^{0}\star a-(-1)^{|a|}a\star P^{0}.\label{eq:d_S}
\end{equation}
Moreover, $P^{1}$ is a cocycle with respect to this new differential,
and we get the following recursive equations for the higher order
terms
\begin{equation}
d_{P^{0}}P^{n}=-\sum_{\underset{i,j\geq1}{i+j=n}}P^{i}\star P^{j}.\label{eq:recursive_MC}
\end{equation}
\end{prop}

\begin{proof}
The Maurer-Cartan equation at order zero in $\hbar$ reads 
\[
dP^{0}+P^{0}\star P^{0}=0,
\]
which means that $P^{0}$ is itself a Maurer-Cartan element. Now it
is a general fact that a differential $d$ twisted by a Maurer-Cartan
element as in (\ref{eq:d_S}) is again a differential. 

The Maurer-Cartan equation at order $1$ in $\hbar$ reads
\[
dP^{1}+P^{0}\star P^{1}+P^{1}\star P^{0}=0,
\]
which is exactly $d_{P^{0}}P^{1}=0$ because $P^{1}$ is of degree
$1$ (it has only one group variable). At last, we obtain (\ref{eq:recursive_MC})
by looking at the MC equation at order $n\geq2$. 
\end{proof}

\section{Existence and rigidity Theorem\label{sec:Existence-and-rigidity}}

In this section, we give cohomological conditions for the existence
of formal $G$-systems, that is, Maurer-Cartan elements in $\mathcal{P}_{\varphi}$.
The discussion that follows is based on appendix A of \cite{ACD}.
The main fact is that $\mathcal{P}_{\varphi}^{\bullet}$ is a \textit{complete}
DGA in the sense of \cite{ACD}; complete DGA have neat cohomological
conditions governing the existence and obstruction of Maurer-Cartan
elements.
\begin{defn}
We define $\pol_{d}(n)$ for $n\geq0$ to be the space of polynomial
in $\xi$ of the form
\[
P(x,\xi)=\sum_{|\alpha|\leq d}f_{\alpha}(x)\xi^{\alpha},
\]
where $f_{\alpha}\in C^{\infty}(\R^{d}).$
\end{defn}
First of all, $\mathcal{P}_{\varphi}^{\bullet}$ has a natural filtration
\[
\cdots\subset F^{k+1}\mathcal{P}_{\varphi}^{\bullet}\subset F^{k}\mathcal{P}_{\varphi}^{\bullet}\subset\cdots\subset F^{1}\mathcal{P}_{\varphi}^{\bullet}\subset F^{0}\mathcal{P}_{\varphi}^{\bullet}=\mathcal{P}_{\varphi}^{\bullet},
\]
for which each of the $(F^{k}\mathcal{P}_{\varphi}^{\bullet},d)$
is a subcomplex and such that
\[
\star:F^{k}\mathcal{P}_{\varphi}^{\bullet}\times F^{l}\mathcal{P}_{\varphi}^{\bullet}\rightarrow F^{k+l}\mathcal{P}_{\varphi}^{\bullet}
\]
This filtration is given by 
\[
F^{k}\mathcal{P}_{\varphi}^{\bullet}=\big\{\sum_{n\geq k}h^{n}P^{n}:\: P^{n}\in\pol_{d}(n)\,\big\},\qquad k\geq1.
\]
We have then a tower
\[
\mathcal{P}_{\varphi}^{\bullet}/F^{1}\mathcal{P}_{\varphi}^{\bullet}\leftarrow\mathcal{P}_{\varphi}^{\bullet}/F^{1}\mathcal{P}_{\varphi}^{\bullet}\leftarrow\cdots,
\]
whose inverse limit is exactly $\mathcal{P}_{\varphi}^{\bullet}$.
This makes $\mathcal{P}_{\varphi}^{\bullet}$ a \textbf{complete}
DGA in the sense of the Appendix A of \cite{ACD}. 

\begin{defn}
Define the graded vector space $\pol_{d}^{\bullet}(n)$ to be 
\[
\pol_{d}^{k}(n):=\big\{ P:G^{k}\rightarrow\pol_{d}(n)\big\},\quad n,k\geq0.
\]
\end{defn}

Observe that, as graded vector space, we have that 
\begin{equation}
\mathcal{F}^{n}\mathcal{P}_{\varphi}^{\bullet}/\mathcal{F}^{n+1}\mathcal{P}_{\varphi}^{\bullet}\;\simeq\;\pol_{d}^{\bullet}(n)\label{eq:identification}
\end{equation}
and the following decomposition of the complex of formal $G$-amplitudes:
\[
\mathcal{P}_{\varphi}^{\bullet}=\pol_{d}^{\bullet}(0)\oplus\hbar\pol_{d}^{\bullet}(1)\oplus\hbar^{2}\pol_{d}^{\bullet}(2)\oplus\cdots
\]
Let $P^{0}\in\pol_{d}^{\bullet}(0)$ be a Maurer-Cartan element. Then
the twisted differential $d_{P^{0}}$ defined by formula (\ref{eq:d_S}),
respects this decomposition and $(\pol_{d}^{\bullet}(n),d_{P0})$
is a complex for each $n\geq0$. These complexes will be the main
ingredients in our existence and rigidity results for formal $G$-systems.

From Proposition \ref{prop:first_and_second_terms}, we get that if
\begin{equation}
P^{0}+hP^{1}+h^{2}P^{2}+\cdots\label{eq:MC_element}
\end{equation}
is a Maurer-Cartan element, then $P^{0}$ is a Maurer-Cartan element
in $\mathcal{P}_{\varphi}^{\bullet}$ and $P^{1}$ is a $1$-cocyle
in $(\pol_{d}^{\bullet}(1),d_{P^{0}})$. Now if we start with a Maurer-Cartan
element $P^{0}$ and and a $1$-cocyle $P^{1}$, in general $P^{0}+hP^{1}$
is not a Maurer-Cartan element in $\mathcal{P}_{\varphi}^{\bullet}$,
and we may wonder whether it is possible to find higher terms to get
a Maurer-Cartan element. 

Another question is whether the representation obtained from (\ref{eq:MC_element})
is equivalent to the one obtained by the first term only, i.e. when
a Maurer-Cartan element is gauge equivalent to its first term. 

\begin{defn}
A Maurer-Cartan element $P^{0}$ in $\mathcal{P}_{\varphi}^{\bullet}$
is called rigid if all Maurer-Cartan elements having as first term
$P^{0}$ are gauge equivalent to this first term. 
\end{defn}

In term of the induced representations, $P^{0}$ being rigid means
that all the representations obtained from Maurer-Cartan elements
of the form (\ref{eq:MC_element}) are equivalent as representations
to the representation
\[
T_{g}^{P^{0}}\psi(x)=P^{0}(x)\psi(\varphi_{g}^{-1}(x)).
\]

The following theorem gives cohomological conditions answering the
questions mentioned above.

\begin{thm}
\label{thm:ExistenceAndRigidity}Let $P^{0}\in\pol_{d}^{\bullet}(0)$
be a Maurer-Cartan element and $P^{1}\in\pol_{d}^{1}(1)$ a one cocycle
(i.e $d_{P^{0}}P^{1}=0$). If 
\[
H^{2}(\pol_{d}^{\bullet}(n),d_{P^{0}})=0,\quad n\geq2,
\]
then there exists a Maurer-Cartan element $\omega$ in $\mathcal{P}_{\varphi}^{\bullet}$
such that 
\[
\omega=P^{0}+\hbar P^{1}+\mathcal{O}(\hbar^{2}).
\]
Moreover if 
\[
H^{1}(\pol_{d}^{\bullet}(n),d_{P^{0}})=0,\quad n\geq1,
\]
the Maurer-Cartan element $P^{0}$ is rigid. 
\end{thm}

\begin{proof}
The proof relies on Proposition A.3 and A.6 of the Appendix A of \cite{ACD}.
Since $\gamma=P^{0}+\hbar P^{1}$ is a Maurer-Cartan element modulo
$\mathcal{F}^{2}\mathcal{P}_{\varphi}^{\bullet}$, Proposition A.3
tells us that there exist a Maurer-Cartan element $\omega=P^{0}+\hbar P^{1}+\mathcal{O}(\hbar^{2})$
provided
\[
H^{2}(\mathcal{F}^{n}\mathcal{P}_{\varphi}^{\bullet}/\mathcal{F}^{n+1}\mathcal{P}_{\varphi}^{\bullet},\, d_{\gamma})=0,\quad n\geq2,
\]
where $d_{\gamma}$ is the operator $d_{\gamma}a=da+[\gamma,a]$,
which becomes a differential on the quotient $\mathcal{F}^{n}\mathcal{P}_{\varphi}^{\bullet}/\mathcal{F}^{n+1}\mathcal{P}_{\varphi}^{\bullet}$.
The first part of the theorem follows from (\ref{eq:identification})
and the fact that $d_{\gamma}$ becomes $d_{P^{0}}$ when passing
to the quotient (because $P^{1}$ has one power of $\hbar$, which
will make this term disappear in the quotient). The rigidity part
of the theorem is a direct application of Proposition A.6 with the
same observations as above.
\end{proof}

\subsection{Trivial action}

Consider the case when group action $G$ is trivial $\varphi_{g}=\id$
as well as the first term of the deformation, that is we are looking
at $G$-systems of the form
\[
a_{g}=1+\hbar P_{g}^{1}+\hbar^{2}P_{g}^{2}+\cdots
\]
The corresponding operators implementing the representation are then
deformations of the trivial representations of $G$ in $C^{\infty}(\R^{d})[[\hbar]]$;
they are of the form 
\begin{equation}
T_{g}^{a}\psi(x)=\id+\sum_{n\geq1}\hbar^{n}P_{g}^{n}\left(x,D\right),\label{eq:def_of_trivial_rep}
\end{equation}
where $P_{g}^{n}(x,D)$ is a differential operator of order $n$ with
nonconstant bounded coefficients. The following theorem gives a simplification
of the existence and rigidity result for general deformations. This
result is very close to that of Pinzcon \cite{Pinzcon} on obstructions
and rigidity of deformations of representations. 

\begin{thm}
\label{thm:trivial-action}Consider the cohomology $H^{\bullet}(G,C^{\infty}(\R^{d}))$
of $G$ with coefficients in the smooth functions on $\R^{d}$, which
we consider as a trivial $G$-bimodule. If $H^{2}(G,C^{\infty}(\R^{d}))=0$
then there exists representation of $G$ into $C^{\infty}[[\hbar]]$
of the form (\ref{eq:def_of_trivial_rep}). Moreover, if $H^{1}(G,C^{\infty}(\R^{d}))=0$,
all these representations are equivalent. 
\end{thm}

\begin{proof}
As a graded vector space $\pol_{d}^{\bullet}(n)$ can be identified
with the following direct sum with $n$-terms 
\[
C^{\bullet}(G,C^{\infty}(\R^{d}))\oplus\cdots\oplus C^{\bullet}(G,C^{\infty}(\R^{d})),
\]
since, for a cochain $P=\sum_{|\alpha|\leq n}f^{\alpha}(x)\xi^{\alpha}$,
we have that $f^{\alpha}\in C^{\bullet}(G,C^{\infty}(\R^{d}))$ for
all multi-indices $\alpha$. Using Theorem \ref{thm:ExistenceAndRigidity},
we only need to show that, in the case the action is trivial, $d_{1}$
respects this splitting. Let us compute the differential of $P\in\pol_{d}^{k}(n)$:
\begin{eqnarray*}
(d_{1}P)_{g_{1},\dots,g_{k+1}} & = & (dP)_{g_{1},\dots,g_{k+1}}+1_{g_{1}}\star P_{g_{2},\dots,g_{k+1}}-(-1)^{k}P_{g_{1},\dots,g_{k}}\star1_{g_{k+1}},\\
 & = & P_{g_{2},\dots,g_{k+1}}+\sum_{i=1}(-1)^{i}P_{g_{1},\dots,g_{i}g_{i+1},\dots,g_{k+1}}+(-1)^{k+1}P_{g_{1},\dots,g_{k}},
\end{eqnarray*}
since the product $\star$ is now the standard product (associated
with the standard quantization) because the action is trivial. Since
only the $f^{\alpha}$'s depend on the group variables, we obtain
that
\[
d_{1}P=\sum_{|\alpha|\leq n}(\tilde{\delta}f^{\alpha})\xi^{\alpha},
\]
where $\tilde{\delta}$ is the differential of the group cohomology
of $G$ in $S_{d}(1)$ considered as a trivial bimodule. 
\end{proof}


\begin{thebibliography}{References}


\bibitem{ACD} C. A. Abad, M. Crainic, and B. Dherin, Tensor products
of representations up to homotopy, \textit{J. of Homotopy and Rel.
Structures} \textbf{6 }(2011), 239\textendash{}-288.

\bibitem{AgAmLeOv}B. Agrebaoui, F. Ammar, P. Lecomte, and V. Ovsienko,
Multi-parameter deformations of the module of symbols of differential
operators, \textit{Int. Math. Res.} \textbf{16} (2002), 847--869.

\bibitem{A1} D. Arnal, \textasteriskcentered{} products and representations
of nilpotent groups, \textit{Pacific J. Math.} \textbf{114} (1984),
285\textendash{}-308. 

\bibitem{A2} D. Arnal, N. Dahmene, K. Tounsi, Poisson action and
formality, \textit{Lett. Math. Phys.} \textbf{82} (2007), 177--189.

\bibitem{As} A. Astashkevich, R. Brylinski, Non-local equivariant
star product on the minimal nilpotent orbit, \textit{Adv. Math. }\textbf{171}
(2002), 86--102.

\bibitem{BFFLS}F. Bayen, M. Flato, C. Fronsdal, A. Lichnerowicz,
and D. Sternheimer, Deformation theory and quantization, I and II,
\textit{Ann. Phys.} \textbf{111} (1977), 61\textendash{}151.

\bibitem{B1} P. Bieliavsky and V. Gayral, Deformation Quantization
for Actions of Kahlerian Lie Groups Part I: Fr\`echet Algebras, arXiv:1109.3419
(2011). 

\bibitem{B2}P. Bieliavsky and M. Massar, Strict deformation quantization
for actions of a class of symplectic Lie groups, \textit{Noncommutative
Geometry and String Theory} (Yokohama, 2001). Progr. Theoret. Phys.
Suppl. \textbf{144} (2001), 1--21.

\bibitem{D} H. Duistermaat, \emph{Fourier Integral Operators}, Progress
in Mathematics \textbf{30}, Birkh\"auser (1996).

\bibitem{EZ}M. Zworski, \textit{Semiclassical Analysis}, Graduate
Studies in Mathematics \textbf{138}, AMS (2012).

\bibitem{GS}V. Guillemin and S. Sternberg, Geometric quantization
and multiplicities of group representations\textit{,} \textit{Invent.
Math.} \textbf{67} (1982), 515--538. 

\bibitem{Ho} L. H\"ormander\emph{, The Analysis of Linear Partial
Differential Operators III},\emph{ Pseudodifferential Operators},
Springer-Verlag (2007). 

\bibitem{KI}A. A. Kirillov, The orbit method. I. Geometric quantization.
Representation theory of groups and algebras, \textsl{Contemp. Math.}\textbf{
145} (1993), 1--32. 

\bibitem{Ko} B. Kostant, \textit{Quantization and Unitary Representations},
Lect. Notes in Math. \textbf{170}, Springer (1970), 87--208. 

\bibitem{M}A. Martinez, \textsl{An Introduction to Semiclassical
and Microlocal Analysis}, Springer (2001). 

\bibitem{NiRi}A. Nijenhuis, R. W. Richardson, Deformations of homomorphisms
of Lie groups and Lie algebras, \textit{Bull. Amer. Math. Soc.} \textbf{73}
(1967), 175--179.

\bibitem{Ov} V. Ovsienko, C. Roger, Deforming the Lie algebra of
vector fields on $S^{1}$ inside the Lie algebra of pseudodifferential
symbols on $S^{1}$, \textit{Differential topology, infinite-dimensional
Lie algebras, and applications}, Amer. Math. Soc. Transl. \textbf{194}
(1999), 211--226.

\bibitem{Pinzcon}G. Pinczon, Deformations of representations, \textit{Lett.
Math. Phys.} \textbf{1} (1977), 535--544.

\bibitem{Rieffel}M. A. Rieffel, Deformation quantization for actions
of $\bR^{d}$, \textit{Mem. Amer. Math. Soc. }\textbf{106} (1993).

\bibitem{Rieffel2}M. A. Rieffel, Deformation quantization of Heisenberg
manifolds, \textit{Comm. Math. Phys.} \textbf{122} (1989), 531--562.

\bibitem{W}N. Woodhouse, \textit{Geometric Quantization}, Oxford
Mathematical Monographs, The Clarendon Press, Oxford University Press
(1980).

\bibitem{LW} G. Lechner and S. Waldmann, Strict deformation quantization
of locally convex algebras and modules, preprint arXiv:1109.5950 (2011).

\bibitem{Xu}P. Xu, Fedosov \textasteriskcentered{}-products and quantum
momentum maps, \textit{Comm. Math. Phys.} \textbf{197} (1998), 167--197. \end{thebibliography}
\end{document}